\documentclass[11pt,a4paper]{article}

\usepackage[T1]{fontenc}
\usepackage[utf8]{inputenc}
\usepackage[a4paper,margin=1in]{geometry}
\usepackage{lmodern}

\usepackage{graphicx}
\usepackage{caption}
\usepackage{subcaption}
\usepackage{amsmath,amssymb,amsfonts}
\usepackage{amsthm}
\usepackage{xcolor}
\usepackage{booktabs}
\usepackage{multirow}
\usepackage{braket}
\usepackage{tikz}
\usetikzlibrary{arrows.meta,decorations.pathmorphing}
\usepackage{enumitem}
\usepackage{microtype}
\usepackage{acronym}
\usepackage[hidelinks]{hyperref}
\usepackage[nameinlink,noabbrev]{cleveref}

\newcommand{\negl}[1]{\operatorname{negl}(#1)}
\DeclareMathOperator{\poly}{poly}
\newcommand{\verifier}{\mathcal{V}}
\newcommand{\prover}{\mathcal{P}}
\newcommand{\adversary}{\mathcal{A}}
\newcommand{\KeyGen}{\mathsf{KeyGen}}
\newcommand{\sample}{\overset{\$}{\leftarrow}}
\newcommand{\Ii}{I_{i}}

\theoremstyle{plain}
\newtheorem{theorem}{Theorem}

\newtheorem{lemma}[theorem]{Lemma}

\theoremstyle{definition}

\newtheorem{definition}{Definition}

\theoremstyle{remark}

\begin{document}

\title{Security Framework for Quantum Distance-Bounding}

\author{
Kevin Bogner$^{1}$, Aysajan Abidin$^{1}$, Dave Singelee$^{2}$, Bart Preneel$^{1}$\\[0.5em]
\small $^{1}$COSIC, KU Leuven, Leuven, Belgium\\
\small $^{2}$DistriNet, KU Leuven, Leuven, Belgium\\
\small \texttt{kevin.bogner@kuleuven.be, aysajan.abidin@kuleuven.be,}\\
\small \texttt{dave.singelee@kuleuven.be, bart.preneel@kuleuven.be}
}

\date{}

\maketitle

\begin{abstract}
Distance-bounding (DB) protocols let a verifier upper-bound a prover's physical
distance by timing rapid challenge-response exchanges. Quantum communication
promises simpler DB protocols with stronger security guarantees, yet existing
quantum distance-bounding (QDB) proposals are analysed in ad-hoc models and, to
the best of our knowledge, lack a common game-based treatment of standard fraud
attacks. We contribute (i) a reusable security framework for QDB that fixes
system and timing assumptions, specifies a quantum-capable adversary model,
formalises distance-, mafia-, and terrorist-fraud experiments, and includes a
simple i.i.d.\ depolarizing noise model; and (ii) an application of this
framework to a published QDB protocol. For this protocol we characterise the
honest per-round acceptance probability under noise and lift it to the
multi-round setting, yielding explicit completeness guarantees as a function of
the number of fast rounds, the acceptance threshold, and the noise parameter.
For active adversaries we bound the per-round success probability of
distance-fraud attacks and analyse the best known mafia-fraud strategy,
deriving corresponding multi-round soundness bounds. We also show that the
protocol is inherently insecure against terrorist-fraud in our model. The
framework cleanly separates protocol-independent definitions from
protocol-specific analysis and can be used to evaluate existing and future QDB
protocols on a common basis.
\end{abstract}


\section{Introduction}
Verifying that a remote device is within a claimed physical distance is a fundamental primitive for access control, secure ranging, and secure localisation. \emph{Distance-bounding} (DB) protocols achieve this by measuring the round-trip time of challenge-response exchanges, allowing the verifier to upper-bound its physical distance to the prover~\cite{brands1993distance}. Typical applications include contactless authentication (e.g.\ RFID and NFC tokens), keyless entry systems, and secure ranging and localisation in wireless networks~\cite{hancke2005rfid}. Classical DB designs resist both \emph{distance-fraud} (DF), where a far-away dishonest prover cheats on its location, and \emph{mafia-fraud} (MF), a relay attack by two external adversaries. The strongest threat is \emph{terrorist-fraud} (TF), where a dishonest prover actively collaborates with a nearby helper. Classical countermeasures against TF remain expensive in terms of additional hardware or bandwidth, especially for low-cost devices~\cite{hancke2011design}.

At a high level, a DB protocol proceeds in three phases. In a \emph{slow} (initialisation) phase, which is not time-critical, the verifier and prover exchange nonces, derive session-specific secrets from long-term keys, and set up all parameters needed for the distance test. This is followed by a \emph{fast} phase consisting of $n$ time-critical rounds: in each round the verifier sends a fresh challenge, the prover must respond immediately, and the verifier measures the round-trip time to infer an upper bound on the distance. Finally, in the \emph{decision} phase, the verifier checks that each accepted round is both timely and logically consistent with the protocol rules, and accepts the session if the number of accepted rounds is at least a threshold~$\tau$. We will use this terminology (slow / fast / decision phase) for both classical and quantum distance-bounding protocols throughout the paper.

Quantum communication promises simpler and stronger DB because non-orthogonal quantum states cannot be perfectly copied~\cite{wootters1982single}, preventing relay adversaries from pre-computing answers. Several \emph{quantum distance-bounding} (QDB) protocols have been proposed in recent years, starting with the protocol of Abidin \emph{et al.}~\cite{abidin2016towards} and its follow-ups~\cite{abidin2019quantum,abidin2024entanglement,bogner2024entangled,bogner2025continuous}. On the classical side, DB protocols have been studied in a unified, game-based framework that fixes system assumptions, threat experiments, and soundness notions~\cite{boureanu2015practical}. In contrast, existing QDB proposals are analysed in protocol-specific models and, to the best of our knowledge, none comes with a formal security proof against standard distance-, mafia-, and terrorist-fraud attacks. We aim to close this gap for QDB in this paper.

In this work we take a step towards such a foundation. We isolate a minimal set of system and timing assumptions for QDB, specify a quantum-capable adversary model, formalise DF/MF/TF experiments in this setting, and connect per-round cheating probabilities to multi-round security. We then apply this framework to a concrete QDB protocol~\cite{abidin2019quantum} and obtain explicit completeness and soundness guarantees as functions of the number of fast rounds~$n$, the threshold~$\tau$, and a simple noise parameter.

For context, we briefly compare the analysed protocol with the classical Hancke-Kuhn style DB protocol as a baseline~\cite{hancke2005rfid}.

The protocol we analyse~\cite{abidin2019quantum} follows the same high-level structure as the classical Hancke-Kuhn DB protocol~\cite{hancke2005rfid}, consisting of a slow setup phase, $n$ time-critical fast rounds, and a final decision that accepts iff at least $\tau$ rounds satisfy the time-of-flight bound and the protocol's value check. The key difference lies in the fast-phase workload
and the resulting system assumptions. Classical Hancke-Kuhn-style DB exchanges
classical challenge and response bits and therefore requires primarily tight
latency engineering on a classical channel. In contrast, QDB replaces the
fast-phase bit exchange by
single-qubit transmissions, which introduces the cost of quantum state preparation
and measurement (and a quantum channel), while leaving the slow phase and its
symmetric-key structure essentially unchanged. Moreover, completeness and parameter
sizing must account for quantum noise on the fast-phase transmissions (captured by
our noise model in~\Cref{subsec:noise-model}). From a security perspective, both
settings rely on time-of-flight, fresh challenges, and symmetric-key assumptions
in the slow phase, but QDB additionally exploits quantum-mechanical constraints
such as no-cloning and measurement disturbance~\cite{wootters1982single}. In our
setting, the per-round secrets are derived via a \emph{quantum-secure} PRF and we
model a quantum-capable (QPT) adversary controlling both the classical and quantum
channels (see~\Cref{sec:model,sec:threats}).

For a simple quantitative comparison, we use Hancke-Kuhn~\cite{hancke2005rfid} as the closest classical baseline, since both protocols have the same slow/fast/decision structure and a lightweight symmetric-key slow phase. To keep this discussion simple at this stage, \Cref{tab:db-qdb-comparison} fixes an idealised noiseless setting with strict threshold $\tau=n$. The table reports the smallest number of rounds needed to reach a session false-accept probability at most $2^{-80}$. The full threshold/noise trade-off is analysed later in~\Cref{sec:security_analysis}.

\begin{table}[t]
\centering
\caption{Comparison with the classical Hancke-Kuhn~\protect\cite{hancke2005rfid} baseline in an idealised noiseless setting with strict threshold $\tau=n$. For Hancke-Kuhn, the standard $3/4$ DF/MF per-round baseline is summarised in~\protect\cite{hancke2011design}. $P_{\mathrm{FA}}$ denotes the false-acceptance probability. The QDB mafia-fraud row uses the best known mafia-fraud attack~\protect\cite{verschoor2022quantum}.}
\label{tab:db-qdb-comparison}
\begin{tabular}{lcc}
\toprule
Metric & Hancke-Kuhn DB~\cite{hancke2005rfid} & Abidin 2019 QDB~\cite{abidin2019quantum} \\
\midrule
Fast-phase transmission & 1 classical bit each way & 1 qubit each way \\
Per-round DF & $3/4$ & $1/2$ \\
Per-round MF & $3/4$ & $7/8$ \\
Rounds for DF, $P_{\mathrm{FA}}\le 2^{-80}$ & $193$ & $80$ \\
Rounds for MF, $P_{\mathrm{FA}}\le 2^{-80}$ & $193$ & $416$ \\
\bottomrule
\end{tabular}
\end{table}

The comparison shows that the present QDB protocol~\cite{abidin2019quantum} should not be viewed as a drop-in replacement for classical DB. While QDB improves the per-round distance-fraud bound, its best known mafia-fraud success probability is less favourable and the fast phase requires quantum hardware. Still, QDB remains worth studying. Quantum communication changes the fast-phase attack surface, and earlier QDB work already argued that early-detect/late-commit strategies do not transfer directly, since an adversary that does not know the correct basis cannot measure the challenge without disturbing the state~\cite{abidin2016towards}. At the same time, QDB introduces its own implementation assumptions and attack surface, as illustrated by photon-number-splitting attacks; indeed, the protocol of Abidin~\cite{abidin2019quantum} was proposed as an improvement over the 2016 variant in this respect, and later entanglement-based proposals continue to explore this design space~\cite{abidin2024entanglement}. We therefore do not claim that the present QDB protocol already dominates classical DB on all metrics; rather, QDB offers additional physical constraints and a different design space that may enable stronger future protocols, which motivates a common evaluation framework.

To make the guarantees of our framework interpretable and comparable across
protocols, we state upfront the system assumptions and modelling choices on which
the main results rely. These assumptions are formalised in~\Cref{sec:model,sec:threats} and are used consistently throughout the protocol
analysis in~\Cref{sec:security_analysis}.

\emph{Timing model.}
As in classical DB, our security guarantees rely on an explicit
time-of-flight constraint. In each fast round, the verifier timestamps the moment
a challenge state is emitted and accepts the round only if the corresponding
response is received within a strict deadline derived from the distance bound
(see~\Cref{subsec:timing-constraint}). Any fixed device latencies, such as state preparation and measurement delays, are assumed to be either negligible or accounted for by calibration and absorbed into the effective distance bound. This enforcement is necessary as it implies that any response generated entirely outside the verifier's proximity and received by the deadline cannot causally depend on the fresh challenge of that round. This property is a direct consequence of no-superluminal signalling (formalised as Lemma~\ref{lem:oob-general}).

\emph{Adversary model and channel control.}
We consider a quantum-capable adversary that controls both the classical and quantum channels, subject only to the constraints of quantum mechanics and relativity. The precise capabilities used in the DF/MF/TF experiments are formalised later in~\Cref{sec:threats}.

\emph{Noise assumptions.}
To capture imperfect quantum communication in \emph{honest} executions, we adopt a simple, protocol-agnostic noise model for the quantum channel (see~\Cref{subsec:noise-model}). This model is used only for analysing honest executions and parameter sizing. In contrast, our security bounds are derived against an \emph{idealised noiseless adversary} with perfect channels and measurements. This is a conservative choice, as any physical imperfections affecting the attacker can only reduce its ability to return correct, timely responses under the same decision rule.

\emph{Threat scope.}
We focus on the standard fraud threats from classical DB, namely distance fraud, mafia fraud, and terrorist fraud, and formalise them via game-based experiments (\Cref{sec:threats}). Our analysis is protocol-level and does not aim to model implementation-specific side channels or physical-layer attacks such as multi-photon effects, detector blinding, device fault attacks, or adversarially induced losses, nor do we treat alternative threat variants such as distance hijacking. These aspects are orthogonal to the core framework and are best addressed by extending the system model.

Our contributions are as follows.
\begin{itemize}
  \item We provide a reusable, game-based security framework for proving the security of QDB protocols, covering system and timing assumptions, quantum adversary models, soundness definitions, and a simple noise model.
  \item We apply this framework to an already proposed QDB protocol~\cite{abidin2019quantum} and derive explicit bounds on completeness and on DF/MF soundness in the multi-round setting.
\end{itemize}

The remainder of the paper is organised as follows. In \Cref{sec:related-work} we position our contribution within the existing literature on classical DB, QDB, and related quantum position verification (QPV). In \Cref{sec:prelims} we introduce the notation for quantum states and the probabilistic tools used later. \Cref{sec:model} formalises QDB protocols, and \Cref{sec:threats} introduces the threat model and security experiments. \Cref{sec:security_analysis} revisits the QDB protocol of Abidin~\cite{abidin2019quantum} and applies the framework to it. \Cref{sec:conclusion} concludes the paper.

\section{Related work}
\label{sec:related-work}

DB was introduced by Brands and Chaum as a cryptographic primitive for upper-bounding the physical distance to a prover through time-of-flight measurements in rapid challenge-response rounds~\cite{brands1993distance}. Since then, classical DB has produced lightweight designs for constrained devices, most notably Hancke-Kuhn-style protocols, and a large body of work on standard fraud threats, namely DF, MF, and TF, under tight latency and noisy-channel constraints~\cite{hancke2005rfid,hancke2011design}. On the formal side, classical DB has also benefited from reusable game-based security frameworks that make the system assumptions explicit and define DF/MF/TF experiments in a uniform way~\cite{boureanu2015practical}. Our QDB framework is derived from this line of work and adapts it to the quantum setting.

Within this landscape, QDB replaces the classical fast-phase bit exchange by quantum communication, typically to exploit no-cloning and measurement disturbance as additional constraints on adversaries~\cite{wootters1982single}. Early work by Abidin \emph{et al.} studied the feasibility of this approach in Hancke-Kuhn-style protocols using BB84-type states and discussed timing, hardware-delay considerations, and informal DF/MF/TF success estimates~\cite{abidin2016towards}. Abidin later refined this line with an updated QDB protocol and a discussion of implementation-driven attacks such as photon-number splitting~\cite{abidin2019quantum}. A related hybrid design by Abidin revisited an earlier qubit-based relay-attack-detection scheme and proposed an improved timed protocol with classical-bit challenges and qubit responses~\cite{abidin2020detecting}. More recent proposals extend the design space to entanglement-based, mutual, and continuous-variable variants~\cite{abidin2024entanglement,bogner2024entangled,bogner2025continuous}.

Despite this progress on protocol design, QDB security analyses remain largely protocol-specific. Existing works typically adopt bespoke system models and analyse selected attack strategies, rather than a common game-based treatment of DF/MF/TF with a standard lifting from per-round to multi-round security. Verschoor's re-analysis illustrates this gap by identifying stronger mafia-fraud strategies for published QDB protocols and by showing that TF resistance can depend subtly on the exact experiment definition~\cite{verschoor2022quantum}. To the best of our knowledge, QDB still lacks a unified treatment analogous to the classical DB frameworks~\cite{boureanu2015practical}.

A related, but distinct, line of work is \emph{quantum position verification} (QPV). QPV uses a party's geographical position as the credential and typically involves multiple verifiers under relativistic constraints, whereas QDB studies shared-key proximity authentication between a verifier and a prover. We therefore do not adopt QPV security notions directly; we only note that, as in QPV, security depends critically on making the timing assumptions and adversary resources explicit~\cite{buhrman2014position}.

\section{Preliminaries}
\label{sec:prelims}

In this section we collect the necessary preliminaries for the security analysis.
We start by introducing the notation for quantum states and measurements, which are inspired by the BB84 QKD protocol~\cite{bennett1984bb84}, and conclude with a concentration bound that we repeatedly use in later sections. The QDB protocol~\cite{abidin2019quantum} that we will apply the framework to is recalled in \Cref{subsec:qdb-protocol}.

\subsection{Notation for quantum states and measurements}\label{subsec:prelims:states-and-measurements}
Let $a\in\{0,1\}$ index the basis ($a=0$ is the computational $Z$ basis, $a=1$ is the diagonal $X$ basis).
Define the four BB84 states
\[
  \ket{0}_{0}=\ket{0},\quad \ket{1}_{0}=\ket{1},\qquad
  \ket{0}_{1}=\ket{+},\quad \ket{1}_{1}=\ket{-},
\]
where
\[
  \ket{\pm}=\tfrac{1}{\sqrt{2}}(\ket{0}\pm\ket{1}).
\]
For $r\in\{0,1\}$ we write the rank-one projectors as
\[
  \Pi_{a,r} := \ket{r}_{a}\bra{r}_{a}.
\]
We write $\bra{r}_{a}:=(\ket{r}_{a})^{\dagger}$ for the associated bra.
We also write $\Pi_{\psi}:=\ket{\psi}\bra{\psi}$ for the projector onto a pure state $\ket{\psi}$.
Measuring $\ket{\psi}$ in basis $a$ yields outcome $r$ with probability
\[
  \Pr[r]=\bigl\lVert\Pi_{a,r}\ket{\psi}\bigr\rVert^2=\bra{\psi}\Pi_{a,r}\ket{\psi}.
\]

\subsection{Concentration bound}\label{subsec:concentration}

We use a single tail bound for our QDB protocol, to bound the success probability
of the protocol in an honest or malicious setting. The key requirement is
\emph{adaptivity-robustness}: even if the adversary chooses its action in
round~$i$ \emph{after seeing everything that happened so far}, the bound still
applies as long as that round's conditional success probability is capped.
For $u,v\in(0,1)$, let
\[
  D(u\Vert v):=u\ln\frac{u}{v} + (1-u)\ln\frac{1-u}{1-v}
\]
denote the binary relative entropy between
$\mathrm{Bernoulli}(u)$ and $\mathrm{Bernoulli}(v)$. In our applications, $u$ will be the
threshold ratio $\tau/n$ (number of required accepted rounds $\tau$ divided by the total number of rounds $n$),
and $v$ will be a bound $p$ on the per-round success probability.

This is a Chernoff-type bound for adaptively generated Bernoulli variables: the
indicators $I_i$ may be chosen based on the full history $\mathcal{F}_{i-1}$,
and we only require the conditional success probability
$\mathbb{E}[I_i \mid \mathcal{F}_{i-1}]$ to be bounded by $p$. In the i.i.d.
case it reduces to the standard Chernoff bound.

\begin{lemma}[Chernoff bound under adaptivity~\cite{mitzenmacher2017probability}]\label{lem:bern-kl}
  Fix any process of fast rounds and let $(\mathcal{F}_{i})_{i=0}^{n}$ be the
  filtration where $\mathcal{F}_{i}$ denotes the information available
  \emph{after round $i$ has completed} (equivalently: just before round $i+1$
  begins), with $\mathcal{F}_0$ the information at the start of the fast phase.
  Concretely, $\mathcal{F}_{i}$ contains
  \begin{itemize}
    \item the \emph{public transcript} up to round $i$ (all classical messages,
          timing/acceptance flags of rounds $1,\dots,i$, and any randomness already revealed); and
    \item the adversary's \emph{internal state} (its private classical registers
          and any quantum registers it keeps between rounds, including entanglement with systems
          not yet measured).
  \end{itemize}
  Let $I_i\in\{0,1\}$ be the indicator that round $i$ is accepted, and let
  $S=\sum_{i=1}^n I_i$.
  If the single-round conditional success is capped as
  \[
    \mathbb{E}[\,I_i \mid \mathcal{F}_{i-1}\,] \le p
    \qquad\text{for every } i\in[n],
  \]
  then for any $\tau\in(np,n]$,
  \[
    \Pr[S \ge \tau]
      \;\le\; \exp\Bigl(-\,n\,D\bigl(\tfrac{\tau}{n}\,\Vert\,p\bigr)\Bigr)
      \quad \text{(upper tail)}.
  \]
  Similarly, if $\mathbb{E}[\,I_i \mid \mathcal{F}_{i-1}\,] \ge p$ for all $i$, then
  for any $\tau\in[0,np)$,
  \[
    \Pr[S \le \tau]
      \;\le\; \exp\Bigl(-\,n\,D\bigl(\tfrac{\tau}{n}\,\Vert\,p\bigr)\Bigr)
      \quad \text{(lower tail)}.
  \]
  \end{lemma}

\section{Model for Quantum Distance-Bounding}
\label{sec:model}
Our model adapts the classical DB framework of
Boureanu et al.~\cite{boureanu2015practical} to the
quantum setting; the definitions and security experiments introduced
here and in \Cref{sec:threats} follow their structure closely.

\subsection{Complexity}
Let \(\lambda \in \mathbb{N}\) denote the security parameter, which controls the security-efficiency tradeoff: as \(\lambda\) grows, costs may increase while any adversary's success probability drops to a negligible function in~\(\lambda\).
Throughout this paper,
\begin{itemize}
    \item a \emph{probabilistic polynomial-time} (PPT) algorithm is a classical probabilistic Turing machine whose running time is bounded by a polynomial in the security parameter~\(\lambda\) (and the length of its explicit inputs);
    \item a \emph{quantum polynomial-time} (QPT) algorithm is a quantum algorithm whose running time is bounded by a polynomial in~\(\lambda\) (and the input length);
    \item \(\poly(\cdot)\) denotes an unspecified polynomial; and
    \item a function \(\negl{\cdot}\) is \emph{negligible} if for every polynomial \(p(\cdot)\),
    \(\negl{\lambda} \le 1/p(\lambda)\) for all sufficiently large \(\lambda\).
\end{itemize}

All honest-party algorithms are efficient; classical algorithms are PPT and
quantum algorithms are QPT.

Let $n\in\mathbb{N}$ be the number of time-critical challenge-response rounds
in the \emph{fast phase} of one protocol execution; we treat $n$ as a tunable parameter.

\subsection{Participants}
A verifier~\(\verifier\) and a prover~\(\prover\) interact over a channel fully controlled by a quantum-capable adversary~\(\adversary\).

\begin{definition}[Quantum Distance-Bounding protocol \cite{boureanu2015practical}]
    A QDB protocol is a tuple
    \[
        \mathrm{QDB} \;=\; (\KeyGen, \prover, \verifier, B),
    \]
    where
    \begin{enumerate}
        \item $\KeyGen(1^{\lambda}) \rightarrow x$ is an efficient (classical or quantum) algorithm that outputs a classical key $x \in \{0,1\}^{\lambda}$;
        \item \(\prover(x)\) is an interactive QPT algorithm;
        \item $\verifier$ is an interactive QPT algorithm. We write
        $\verifier(x,r) \rightarrow \mathrm{Out}_{\verifier} \in \{0,1\}$
        for its execution on key $x$ with internal randomness $r$, where
        $\mathrm{Out}_{\verifier}$ is the verifier's final decision bit; and
        \item \(B\) is the maximum allowed distance between $\verifier$ and $\prover$.
    \end{enumerate}
\end{definition}

\subsection{Timing constraint}\label{subsec:timing-constraint}
Let \(c\) be the speed of light (not to be confused with the challenge bits \(c_i\)).
During the fast phase, \(\verifier\) rejects whenever the measured round-trip time exceeds
\[
    \Delta t_{\max} \;:=\; \frac{2B}{c}.
\]
We take double the distance $B$ to account for the round-trip time.

\subsection{Noise model}\label{subsec:noise-model}
We model physical noise by assuming that each transmitted qubit, on both the challenge hop ($\verifier\to\prover$) and the response hop ($\prover\to\verifier$), undergoes an \emph{independent and identically distributed} (\emph{i.i.d.}) depolarizing channel~\cite{nielsen2010quantum} $\mathcal{D}_\eta$, so there are no correlations across the protocol rounds or across the hops. The depolarizing channel with parameter $\eta\in[0,1]$ acts on any single-qubit projector $\Pi$ as
\[
\mathcal{D}_{\eta}(\Pi)=(1-\eta)\Pi+\eta\,\mathbb{I}/2,
\]
i.e., with probability $1-\eta$ the qubit is unchanged, and with probability $\eta$ it is replaced by the \emph{maximally mixed} state $\mathbb{I}/2$ (complete randomness), which yields a \emph{fair coin flip} under measurement in \emph{any} basis.

When the parties measure in the intended bases, a single noisy hop yields the correct bit with probability
\[
1-\eta/2=(1-\eta)\cdot 1+\eta/2.
\]
Across a full fast round (two hops), the final bit equals the original iff either both hops preserve the bit or both flip it, so the honest per-round acceptance is
\[
p(\eta)=(1-\eta/2)^2+(\eta/2)^2=1-\eta+\eta^2/2.
\]

\subsection{Completeness}\label{sec:req}

Completeness is the \emph{liveness} condition: if the prover is honest and physically within the distance bound, then the protocol should accept except with negligible probability. We write $(A \leftrightarrow B)$ for the execution of the interactive protocol between parties $A$ and $B$. We define completeness as follows:

\begin{definition}[Completeness \cite{boureanu2015practical}] \label{def:completeness}
    If an honest prover~\(\prover\) is located at distance \(d \le B\) from the verifier~\(\verifier\), then
    \[
        \Pr\Bigl[
            (\verifier(x,r)\;\leftrightarrow\;\prover(x))
            \text{ accepts}
            \Bigr]
        \;\ge\; 1 - \negl{\lambda}.
    \]
\end{definition}

\section{Threat model}\label{sec:threats}

Throughout this paper, the adversary~\(\adversary\) is QPT and controls both classical and quantum channels: it may intercept, delay, inject, drop, relay, store systems, and apply arbitrary \emph{completely positive, trace-preserving (CPTP)} maps to any quantum data it possesses~\cite{nielsen2010quantum}. A CPTP map is any physically allowed quantum operation on a state. All actions must respect quantum mechanics and relativity (e.g., no cloning, measurement disturbance, no superluminal signalling).

\subsection{Security experiments}\label{sec:threats:experiments}
We follow a standard game-based template. For each threat, we define an \emph{experiment} that specifies the parties, the adversary's capabilities, and the acceptance condition; the experiment outputs $1$ iff the verifier accepts. The adversary's \emph{advantage} is the probability that the experiment outputs $1$. A protocol is \emph{secure} against that threat if every QPT adversary achieves negligible advantage as a function of the security parameter~$\lambda$.

In the following, we define the three main threat classes: distance-fraud (DF), mafia-fraud (MF), and terrorist-fraud (TF).

\subsubsection{Distance-fraud} \label{df}
\begin{definition}[Distance-fraud experiment adapted from \cite{boureanu2015practical}]
    \label{def:df-experiment}
    \noindent
    \begin{enumerate}
        \item \textbf{Setup.}
              Verifier~$\verifier$ and dishonest prover~$\prover^\star$
              share a long-term secret key~$x$. ~$\prover^\star$ is located at distance $d>B$ from~$\verifier$.
        \item \textbf{Challenge session.}
              \(\verifier\) and \(\prover^{\star}\) interact in a complete execution of the QDB protocol, consisting of \(n\) fast rounds obeying the distance bound limit as in an honest run.
        \item \textbf{Output.} $\verifier$ outputs \(\mathrm{Out}_{\verifier}\in\{0,1\}\).
    \end{enumerate}
\end{definition}

\begin{definition}[Distance-fraud advantage adapted from \cite{boureanu2015practical}]
\label{def:df-adv}
    For every QPT dishonest prover~\(\prover^{\star}\) located at distance \(d>B\) from
    verifier~\(\verifier\), define
    \[
        \mathsf{Adv}^{\mathrm{DF}}_{\mathrm{QDB},\prover^\star}(\lambda)
        :=\;
        \Pr\left[
            \begin{array}{l}
                x \leftarrow \KeyGen(1^\lambda); \\
                \text{Out}_{\verifier} \leftarrow
                (\verifier(x,r)\;\leftrightarrow\;\prover^\star(x))
            \end{array}
            :\;
            \text{Out}_{\verifier}=1
            \right].
    \]
\end{definition}

\begin{definition}[Distance-fraud security adapted from \cite{boureanu2015practical}]
\label{def:df}
    A QDB protocol is \(\varepsilon_{\mathrm{DF}}(\lambda)\)-distance-fraud secure
    if, for every QPT adversary \(\prover^\star\),
    \[
        \mathsf{Adv}^{\mathrm{DF}}_{\mathrm{QDB},\prover^\star}(\lambda)
        \;\le\;
        \varepsilon_{\mathrm{DF}}(\lambda),
    \]
    where \(\varepsilon_{\mathrm{DF}}\) is a negligible function.
\end{definition}

\subsubsection{Mafia-fraud}
\label{mf}
\begin{definition}[Mafia-fraud experiment adapted from \cite{boureanu2015practical}]
  \label{def:mf-experiment}
    \noindent
    \begin{enumerate}
        \item \textbf{Setup.}
              Verifier~\(\verifier\) and honest prover~\(\prover\) share a long-term
              secret key~\(x\).
              Two QPT adversaries,
              \(\adversary_1\) (co-located with~\(\verifier\)) and
              \(\adversary_2\) (co-located with~\(\prover\)),
              share an authenticated classical and quantum channel.
        \item \textbf{Learning phase.}
              The adversaries may initiate and control any polynomial number of auxiliary executions of the QDB protocol between $\verifier$ and $\prover$:
              in each such execution, every classical and quantum message between the honest parties passes through $(\adversary_1,\adversary_2)$, who may relay, delay, drop, modify, or inject messages arbitrarily (subject only to the timing constraints).
              At the end of this phase, $(\adversary_1,\adversary_2)$ retain the entire classical transcript and their joint quantum state.
        \item \textbf{Challenge session.}
              The adversaries $(\adversary_1,\adversary_2)$ interact with $\verifier$ in a full execution of the QDB protocol.
              Simultaneously, they may interact with the honest $\prover$ (who uses the correct long-term key $x$ and is at distance $d > B$).
              $\verifier$ enforces the distance bound $B$.
        \item \textbf{Output.} $\verifier$ outputs \(\mathrm{Out}_{\verifier}\in\{0,1\}\).
    \end{enumerate}
\end{definition}

\begin{definition}[Mafia-fraud advantage adapted from \cite{boureanu2015practical}]
\label{def:mf-adv}
    Let \((\adversary_1,\adversary_2)\) be any QPT pair with
    \(\adversary_1\) close to verifier~\(\verifier\) and \(\adversary_2\)
    close to prover~\(\prover\) at distance \(d>B\) from \(\verifier\).
    Set
    \[
        \mathsf{Adv}^{\mathrm{MF}}_{\mathrm{QDB},\adversary_1,\adversary_2}(\lambda)
        :=\;
        \Pr\left[
            \begin{array}{l}
                x \leftarrow \KeyGen(1^\lambda); \\
                (\adversary_1,\adversary_2)\,
                \leftrightarrow\,\prover(x); \\
                \text{Out}_{\verifier} \leftarrow
                \bigl(\,\verifier(x,r)\;\leftrightarrow\;
                \adversary_1 \parallel \adversary_2\,\bigr)
            \end{array}
            :\;
            \text{Out}_{\verifier}=1
            \right].
    \]
\end{definition}

\begin{definition}[Mafia-fraud security adapted from \cite{boureanu2015practical}]
  \label{def:mf}
    A QDB protocol is \(\varepsilon_{\mathrm{MF}}(\lambda)\)-mafia-fraud secure if, for every QPT pair
    \((\adversary_1,\adversary_2)\),
    \[
        \mathsf{Adv}^{\mathrm{MF}}_{\mathrm{QDB},\adversary_1,\adversary_2}(\lambda)
        \;\le\;
        \varepsilon_{\mathrm{MF}}(\lambda),
    \]
    where \(\varepsilon_{\mathrm{MF}}\) is a negligible function.
\end{definition}

\subsubsection{Terrorist-fraud}\label{tf}

Terrorist-fraud models the strongest collusion attack in distance-bounding:
a dishonest prover \(\prover^\star\) located outside the bound \(B\) collaborates
with a nearby helper \(\mathcal{A}\) co-located with \(\verifier\). The goal is for \(\mathcal{A}\) to convince \(\verifier\) during the \emph{fast phase}, even
though \(\prover^\star\) is too far away to respond in time. In contrast to
mafia-fraud, the prover itself is malicious and may deliberately provide
information or quantum systems to the helper in order to pass the challenge
session.

Intuitively, terrorist-fraud lies between mafia-fraud and outright key disclosure. As in mafia-fraud, a nearby helper \(\mathcal{A}\) must answer the verifier during the fast phase; unlike mafia-fraud, however, the distant prover \(\prover^\star\) is itself dishonest and may deliberately assist. This assistance is nevertheless restricted: it may help only for the current session and should not equip \(\mathcal{A}\) to authenticate alone in future sessions. The core TF question is therefore whether \(\prover^\star\) can provide session-specific help that lets \(\mathcal{A}\) answer the fast rounds, yet remains non-transferable in the sense formalised below.

\begin{definition}[Terrorist-fraud experiment adapted from \cite{boureanu2015practical}]
\label{def:tf-experiment}
\noindent
\begin{enumerate}
  \item \textbf{Setup.}
        Verifier~$\verifier$ and dishonest prover~$\prover^\star$ share a long-term secret key~$x$.
        ~$\prover^\star$ is located at distance $d>B$ from~$\verifier$.
        A helper $\mathcal{A}$ is co-located with~$\verifier$.
        The pair $(\prover^\star,\mathcal{A})$ share an authenticated classical and quantum channel.
  \item \textbf{Learning phase.}
        $(\prover^\star,\mathcal{A})$ may engage in any polynomial number of auxiliary executions of the QDB protocol with $\verifier$.
        In each such execution, $\mathcal{A}$ is co-located with $\verifier$ and may relay, modify, or inject messages, while $\prover^\star$ participates at its true distant location using key~$x$.
        The pair may also exchange arbitrary classical and quantum messages over their channel, and all information gathered in this phase is available later.
  \item \textbf{Challenge session.}
        $\mathcal{A}$ impersonates $\prover$ in a full execution of the QDB protocol with $\verifier$,
        subject to the same distance bound~$B$ as in an honest run.
  \item \textbf{Output.} $\verifier$ outputs $\mathrm{Out}_{\verifier}\in\{0,1\}$.
\end{enumerate}
\end{definition}

Without further restriction, \Cref{def:tf-experiment} would be trivial: \(\prover^\star\)
could simply reveal reusable secret information (e.g., the long-term key \(x\)),
after which \(\mathcal{A}\) could impersonate \(\prover\) not only in the current
challenge session but also in future sessions. The distinguishing feature of
terrorist-fraud is therefore that the prover's help is allowed to be useful for
the current session, yet should not give the helper a reusable ability to
authenticate on its own later. We formalise this \emph{non-transferability}
requirement as follows:

\begin{definition}[Non-transferable assistance]\label{def:non-transferable-assistance}
  Let $(\prover^\star,\mathcal{A})$ be a QPT pair as in \Cref{def:tf-experiment}, and fix a long-term key~$x$.
  Run the terrorist-fraud experiment once and consider the final (classical and quantum) state of~$\mathcal{A}$ at the end of this execution.
  Now let $\mathcal{A}$, starting from this state and with no further interaction with $\prover^\star$, engage alone in a second execution of the QDB protocol with $\verifier$, where the long-term key is still $x$ and all nonces and verifier randomness are freshly sampled.
  We say that the assistance of $(\prover^\star,\mathcal{A})$ is \emph{non-transferable} if the probability that $\verifier$ accepts in this second execution is negligible in the security parameter~$\lambda$.
\end{definition}

\begin{definition}[Terrorist-fraud advantage adapted from \cite{boureanu2015practical}]
\label{def:tf-adv}
Let $(\prover^\star,\mathcal{A})$ be any QPT pair as in
\Cref{def:tf-experiment}, whose assistance is non-transferable in the
sense of \Cref{def:non-transferable-assistance}. The terrorist-fraud
advantage of $(\prover^\star,\mathcal{A})$ against QDB is
\[
  \mathsf{Adv}^{\mathrm{TF}}_{\mathrm{QDB},\prover^\star,\mathcal{A}}(\lambda)
  :=\;
  \Pr\left[
    \begin{array}{l}
	      x \leftarrow \KeyGen(1^\lambda); \\
	      \mathrm{Out}_{\verifier} \leftarrow
	      \bigl(\,\verifier(x,r)\;\leftrightarrow\;(\prover^\star(x),\mathcal{A})\,\bigr)
	    \end{array}
	    :\; \mathrm{Out}_{\verifier}=1
	  \right],
\]
where the interaction $(\verifier \leftrightarrow (\prover^\star,\mathcal{A}))$
is the terrorist-fraud experiment of \Cref{def:tf-experiment}.
\end{definition}

\begin{definition}[Terrorist-fraud security adapted from \cite{boureanu2015practical}]
\label{def:tf}
A QDB protocol is $\varepsilon_{\mathrm{TF}}(\lambda)$-terrorist-fraud secure if, for every QPT pair $(\prover^\star,\mathcal{A})$ whose assistance is non-transferable as in \Cref{def:non-transferable-assistance},
\[
\mathsf{Adv}^{\mathrm{TF}}_{\mathrm{QDB},\prover^\star,\mathcal{A}}(\lambda)
\;\le\;
\varepsilon_{\mathrm{TF}}(\lambda),
\]
where $\varepsilon_{\mathrm{TF}}$ is a negligible function.
\end{definition}

\subsubsection{Soundness} \label{subsec:soundness}
\begin{definition}[Soundness w.r.t.\ a set of threats adapted from \cite{boureanu2015practical}]
\label{def:T-soundness}
Let $\mathcal T\subseteq\{\mathrm{DF},\mathrm{MF},\mathrm{TF}\}$.
A QDB protocol is $\mathcal T$-sound if it is secure against all threats in $\mathcal T$,
i.e., for each $T\in\mathcal T$ there exists a negligible function $\varepsilon_T(\lambda)$ such that
$\mathsf{Adv}^{T}_{\mathrm{QDB}}(\lambda)\le \varepsilon_T(\lambda)$.
We call the special case $\mathcal T=\{\mathrm{DF},\mathrm{MF},\mathrm{TF}\}$ \emph{full soundness}.
\end{definition}

\subsection{Framework instantiation across QDB protocol families}

Instantiating our security framework for a QDB protocol follows the same
template throughout. One identifies (i) the messages exchanged in the slow and
fast phases; (ii) the
\emph{per-round acceptance condition}, which always decomposes into a
timing check and a protocol-specific
\emph{value or correlation check}; and (iii) a noise model that determines
the honest per-round acceptance probability. Once a per-round success bound $p$
is known for a given fraud experiment (DF/MF/TF), the same lifting argument via
Lemma~\ref{lem:bern-kl} yields the corresponding multi-round bound.

\emph{Prepare-and-measure single-qubit QDB (Abidin et al.~\cite{abidin2016towards}, Abidin~\cite{abidin2019quantum}).}
These QDB protocols fit our security framework essentially verbatim. The slow phase derives
per-round secrets, each fast round uses a single-qubit challenge and a
single-qubit response, and a round is accepted iff the response is timely and the verifier's
measurement outcome matches the verifier's fresh challenge bit. This directly
matches the DF/MF/TF experiments of~\Cref{sec:threats:experiments} and yields explicit per-round cheating probabilities. The MF bound of Abidin~\cite{abidin2019quantum} was subsequently tightened by Verschoor~\cite{verschoor2022quantum}.

\emph{Entanglement-based QDB (Abidin et al.~\cite{abidin2024entanglement}, Bogner et al.~\cite{bogner2024entangled}).}
Entanglement-based QDB protocols replace the prepare-and-measure fast phase by
operations on shared entangled systems. In our security framework, the
instantiation step is to define the per-round acceptance condition as the conjunction of the timing check and a protocol-specific \emph{correlation check}.
The DF/MF/TF experiments of \Cref{sec:threats:experiments} remain unchanged.
However, Bogner et al.~\cite{bogner2024entangled} provide only an informal security analysis and do not report per-round cheating probabilities under the standard fraud experiments.
Completeness analysis for these QDB protocols typically requires a noise model that
captures entanglement decoherence rather than single-qubit depolarisation.

\emph{Continuous-variable QDB (Bogner et al.~\cite{bogner2025continuous}).}
Continuous-variable proposals use quantum states and measurements with
continuous outcomes in the fast phase. They still instantiate the experiments of \Cref{sec:threats:experiments} by
choosing a per-round acceptance condition consisting of the timing check and a
value check, but the latter usually takes the form of a threshold test on a
continuous measurement outcome. As a result, reporting a single per-round
DF/MF/TF success probability requires specifying the exact
thresholding rule and the corresponding noise model.
In contrast to the single-qubit setting, noise in continuous-variable protocols is often dominated by loss and Gaussian noise rather than depolarisation.

\emph{Summary.}
Across all QDB protocol families, the protocol-independent parts of the security framework,
namely the threat experiments (\Cref{sec:threats:experiments}), the timing model (\Cref{subsec:timing-constraint}), and the multi-round lifting (Lemma~\ref{lem:bern-kl}), remain
unchanged; what varies across proposals is the definition of the per-round value
predicate and the noise model needed to quantify completeness and fraud
success rates.

\section{Security analysis of Abidin's QDB protocol}\label{sec:security_analysis}
In this section we apply the framework to the QDB protocol of Abidin~\cite{abidin2019quantum} (\Cref{subsec:qdb-protocol}). After recalling the protocol, we prove completeness (\Cref{def:completeness}), derive distance-fraud (\Cref{def:df}) and mafia-fraud (\Cref{def:mf}) bounds, and analyse terrorist-fraud (\Cref{def:tf}), highlighting an insecurity. We then establish two-fraud soundness (DF/MF). Each analysis first bounds single-round acceptance and is lifted to the multi-round setting via Lemma \ref{lem:bern-kl}. The noise model (\Cref{subsec:noise-model}) is used for the completeness analysis and threshold sizing; the DF/MF bounds are noise-free. This is a conservative choice, as we assume a perfect adversary (with noiseless channels and measurements), so any additional physical noise can only make DF/MF attacks less effective.

\subsection{Protocol description}\label{subsec:qdb-protocol}

\Cref{fig:QDB} depicts one session of the QDB protocol proposed by Abidin~\cite{abidin2019quantum} between a verifier~$\verifier$ and a prover~$\prover$. These two parties share a long-term key~$x$.
The session consists of an authenticated and public \emph{slow} phase, where nonces and per-round secrets are set up, followed by a \emph{fast} phase of $n$ rapid quantum challenge-response rounds whose round-trip times are measured by~$\verifier$. Finally, in the decision phase, $\verifier$ checks if the response of $\prover$ is consistent with the challenge, if the round-trip time is within the allowed bound, and accepts or rejects the session.

\subsubsection{Slow phase (setup; untimed)}
\begin{enumerate}
  \item \emph{Nonce exchange.} $\verifier$ samples $N_v \sample \{0,1\}^n$ and sends it to $\prover$; $\prover$ samples $N_p \sample \{0,1\}^n$ and sends it back to $\verifier$ (straight arrows in \Cref{fig:QDB} are classical messages). This step is performed to ensure \emph{freshness} of the protocol's execution.
  \item \emph{Deriving per-round secrets.} Using a keyed quantum-secure pseudorandom function (PRF) $f_{x}$~\cite{zhandry2021construct}, and the exchanged nonces, allows both parties to compute a pseudorandom $2n$-bit string $a \parallel b$ (computationally indistinguishable from uniform to any QPT adversary without knowledge of $x$), and divide it into two $n$-bit sequences
  \[
    a \parallel b \;=\; f_{x}(N_v, N_p), \qquad a=(a_1,\dots,a_n),\; b=(b_1,\dots,b_n).
  \]
  The bit $a_i$ selects the \emph{challenge basis} for round~$i$, while $b_i$ selects the \emph{response basis} for round~$i$.
\end{enumerate}

\subsubsection{Fast phase (distance measurement; timed)}
For each round $i=1,\dots,n$:
\begin{enumerate}
  \item \emph{Challenge preparation.} $\verifier$ samples a uniform challenge bit $c_i \sample \{0,1\}$ and prepares the challenge state in the challenge basis $a_i$ (as defined in \Cref{subsec:prelims:states-and-measurements})
  \[
      \ket{\psi_i} = \ket{c_i}_{a_i},
  \]
  and sends it to $\prover$ (wavy arrow denotes a single-qubit transmission in \Cref{fig:QDB}). $\verifier$ records the sending time as $t_i^{\text{send}}$.
  \item \emph{Immediate measurement.} Upon receipt, $\prover$ measures the challenge state $\ket{\psi_i}$ in basis $a_i$, obtaining the measurement result $c'_i$. In a noiseless setting, $c'_i=c_i$.
  \item \emph{Response preparation.} $\prover$ prepares the response state $\ket{\phi_i}$ by preparing the measurement result $c'_i$ using the response basis $b_i$ and sends (as defined in \Cref{subsec:prelims:states-and-measurements})
  \[
      \ket{\phi_i} = \ket{c'_i}_{b_i},
  \]
  back to $\verifier$ (second wavy arrow in \Cref{fig:QDB}).
  \item \emph{Response measurement.} Upon receipt, $\verifier$ measures the response state $\ket{\phi_i}$ in basis $b_i$, obtaining the measurement result $c''_i$. $\verifier$ records the receiving time as $t_i^{\text{recv}}$.
\end{enumerate}

\subsubsection{Decision phase}
After $n$ rounds, $\verifier$ checks for each round $i$ if their measurement result is consistent with the challenge $c_i = c''_i$, and checks if the round-trip time $t_i^{\text{recv}} - t_i^{\text{send}} \le \Delta t_{\max}$ is within the allowed bound. If both conditions are met, $\verifier$ accepts the round $i$. If the number of accepted rounds is at least the threshold $\tau$, $\verifier$ accepts the session. Otherwise, $\verifier$ rejects the session.

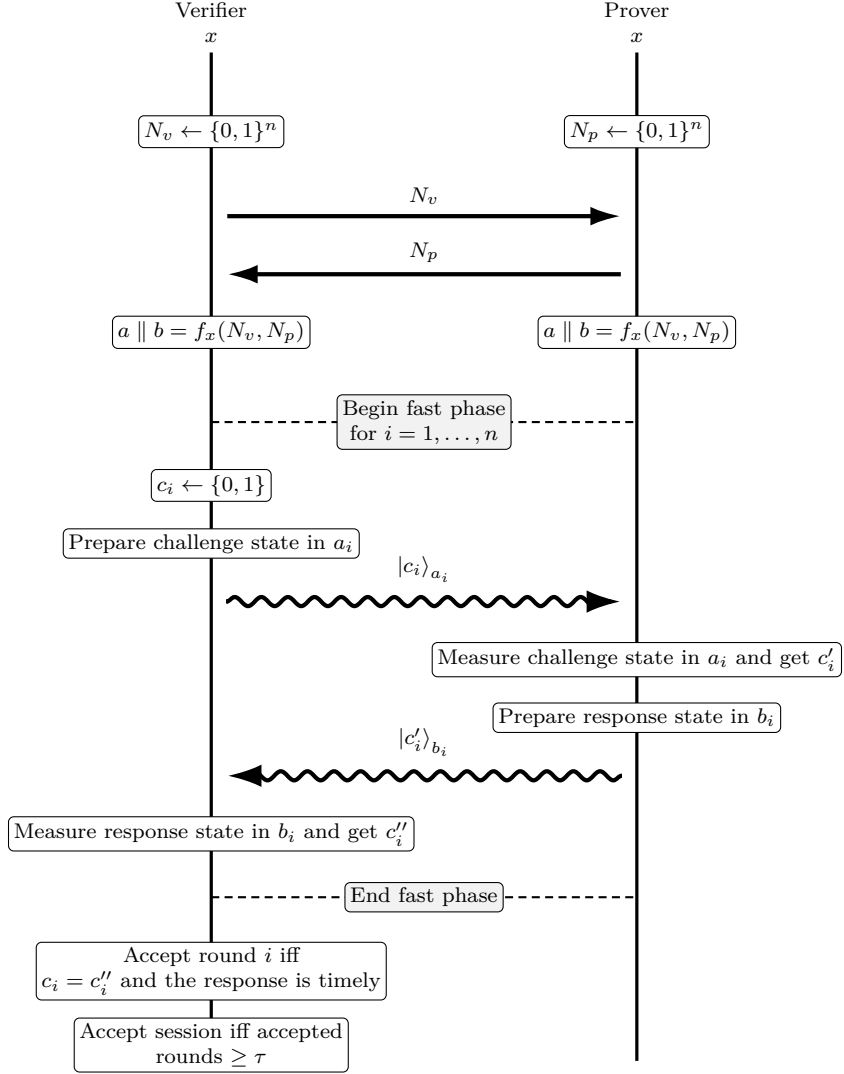
\begin{figure}[htbp]
    \centering
    \begin{tikzpicture}[
        >=Latex,
        x=4cm,
        y=0.7cm,
        font=\scriptsize,
        lifeline/.style={very thick},
        actor/.style={align=center},
        action/.style={draw, fill=white, rounded corners=2pt, align=center, inner sep=2pt},
        condition/.style={draw, fill=gray!10, rounded corners=2pt, align=center, inner sep=2pt},
        msg/.style={-{Latex[length=4mm,width=2.5mm]}, ultra thick, shorten >=6pt, shorten <=6pt},
        qmsg/.style={decorate, decoration={snake, amplitude=0.6mm, segment length=3.5mm, pre length=2mm, post length=3mm}, -{Latex[length=4.5mm,width=3mm]}, ultra thick, shorten >=6pt, shorten <=6pt}
    ]
        \coordinate (Vtop) at (0,0);
        \coordinate (Ptop) at (1.4,0);

        \node[actor, above=0pt] at (Vtop) {Verifier\\$x$};
        \node[actor, above=0pt] at (Ptop) {Prover\\$x$};

        \draw[lifeline] (Vtop) -- (0,-19.1);
        \draw[lifeline] (Ptop) -- (1.4,-19.1);

        \node[action] at (0,-1.5) {$N_v \leftarrow \{0,1\}^n$};
        \node[action] at (1.4,-1.5) {$N_p \leftarrow \{0,1\}^n$};

        \draw[msg] (0,-3.1) -- (1.4,-3.1) node[midway, above] {$N_v$};
        \draw[msg] (1.4,-4.2) -- (0,-4.2) node[midway, above] {$N_p$};

        \node[action] at (0,-5.3) {$a \parallel b = f_{x}(N_v,N_p)$};
        \node[action] at (1.4,-5.3) {$a \parallel b = f_{x}(N_v,N_p)$};

        \draw[densely dashed, thick] (0,-7.0) -- (1.4,-7.0);
        \node[condition] at (0.7,-7.0) {Begin fast phase\\for $i = 1,\dots,n$};

        \node[action] at (0,-8.2) {$c_i \leftarrow \{0,1\}$};
        \node[action] at (0,-9.3) {Prepare challenge state in $a_i$};

        \draw[qmsg] (0,-10.4) -- (1.4,-10.4) node[pos=0.5, above=3pt] {$\ket{c_i}_{a_i}$};

        \node[action] at (1.4,-11.5) {Measure challenge state in $a_i$ and get $c'_i$};
        \node[action] at (1.4,-12.6) {Prepare response state in $b_i$};

        \draw[qmsg] (1.4,-13.7) -- (0,-13.7) node[pos=0.5, above=3pt] {$\ket{c'_i}_{b_i}$};

        \node[action] at (0,-14.8) {Measure response state in $b_i$ and get $c''_i$};

        \draw[densely dashed, thick] (0,-16.0) -- (1.4,-16.0);
        \node[condition] at (0.7,-16.0) {End fast phase};

        \node[action,align=center] at (0,-17.4)
        {Accept round $i$ iff\\
         $c_i=c''_i$ and the response is timely};

        \node[action,align=center] at (0,-18.8)
        {Accept session iff accepted\\
         rounds $\ge \tau$};

    \end{tikzpicture}
    \caption{QDB Protocol~\protect\cite{abidin2019quantum}, based on the classical DB protocol of Hancke-Kuhn~\protect\cite{hancke2005rfid}.}
    \label{fig:QDB}
\end{figure}

\subsection{Completeness}\label{subsec:completeness}

We prove completeness (\Cref{def:completeness}) for the QDB protocol of \Cref{subsec:qdb-protocol}. A fast round is accepted iff the response is timely and the value check passes (i.e., $c_i=c''_i$). In our analysis we separate these two aspects to keep the proofs modular. Let $I_i\in\{0,1\}$ indicate acceptance in round $i$ and $S=\sum_{i=1}^n I_i$ the total number of accepted rounds. The verifier accepts the session iff $S\ge\tau$, where $\tau$ is the protocol's threshold.

\begin{lemma}[Single round completeness]\label{lem:honest-single}
In the ideal (noiseless) model an honest prover passes every round with probability $1$:
\[
\Pr[\Ii=1]=1
\]
for all \(i\in[n]\).
\end{lemma}

\begin{proof}
  \noindent\emph{Timing.}\par
  \noindent Since \(d\le B\), the round-trip time is at most \(2B/c=\Delta t_{\max}\), which is less than the response deadline and therefore the timing bound holds.

  \smallskip
  \noindent\emph{Value.}\par
  \noindent
  Let $\ket{\psi_i}=\ket{c_i}_{a_i}$ be the challenge state and
  $\Pi_{a_i,c_i}=\ket{c_i}_{a_i}\bra{c_i}_{a_i}$ (as in \Cref{subsec:prelims:states-and-measurements}).
  \[
    \Pr[c'_i=c_i]
      = \bigl\lVert\Pi_{a_i,c_i}\ket{\psi_i}\bigr\rVert^2
      = \bigl\lVert\ket{\psi_i}\bigr\rVert^2
      = 1,
  \]
  With $\ket{\phi_i}=\ket{c'_i}_{b_i}$ and
  $\Pi_{b_i,c'_i}=\ket{c'_i}_{b_i}\bra{c'_i}_{b_i}$, the same reasoning yields
  \[
    \Pr[c''_i=c'_i]
      = \bigl\lVert\Pi_{b_i,c'_i}\ket{\phi_i}\bigr\rVert^2
      = \bigl\lVert\ket{\phi_i}\bigr\rVert^2
      = 1,
  \]
  so \(c''_i=c'_i=c_i\) with probability~\(1\), and the value test passes. Hence, \(\Ii=1\).
\end{proof}

In the following, we prove completeness under the noise model (\Cref{subsec:noise-model}) in a multi-round setting (Lemma \ref{lem:bern-kl}).

\begin{theorem}[Completeness under i.i.d.\ depolarizing noise]\label{thm:completeness-noise}
  Let each transmitted qubit pass independently through $\mathcal{D}_{\eta}$ on the forward and return channels. For any threshold $\tau\in\{0,\dots,n\}$ with $\tau<np(\eta)$, where $p(\eta)$ is as in \Cref{subsec:noise-model},
  \[
  \Pr\bigl[S<\tau\bigr] \;\le\; \exp\Bigl(-\,n \cdot D\bigl(\tfrac{\tau}{n}\,\big\Vert\, p(\eta)\bigr)\Bigr).
  \]
  In particular, for $\eta=0$ we have $\Pr[S<\tau]=0$ by Lemma~\ref{lem:honest-single}.
\end{theorem}

\begin{proof}
  Since $d\le B$, the timing bound holds deterministically for an honest prover (Lemma~\ref{lem:honest-single}); thus $I_i$ reduces to the value test in each round. By the noise model (\Cref{subsec:noise-model}), the honest per-round acceptance probability equals $p(\eta)$, hence $\mathbb{E}[I_i \mid \mathcal{F}_{i-1}] = p(\eta)$ for all $i$. Applying Lemma~\ref{lem:bern-kl} (lower tail) with $p=p(\eta)$ yields the claim. In the i.i.d.\ setting, this coincides with the binomial lower-tail bound for $S\sim\mathrm{Bin}(n,p(\eta))$.
\end{proof}

\subsection{Distance-fraud security}\label{subsec:distance-fraud-security}

Recall the distance-fraud experiment from \Cref{def:df-experiment}. The per-round secrets $a=(a_1,\ldots,a_n)$ and $b=(b_1,\ldots,b_n)$ are fixed in the slow phase and therefore \emph{known} to $\prover^{\star}$ who holds the shared secret key $x$. Throughout the fast phase the distant prover $\prover^{\star}$ is located at distance $d>B$. Since $\verifier$ accepts a round only if it receives a response within $\Delta t_{\max}=2B/c$ of sending the challenge, any response system received by the deadline from distance $d>B$ must have been emitted before a light-speed signal carrying $\ket{\psi_i}$ could reach the emitter. Thus, any such timely response is independent of the fresh challenge bit $c_i$. This is formalized in the following Lemma, which is adapted from \cite{boureanu2015practical}.

\begin{lemma}[Out-of-bound independence adapted from \cite{boureanu2015practical}]
    \label{lem:oob-general}
    Fix a fast-phase round $i$. Let $\mathsf{Near}$ denote the ball of radius $B$
    around the verifier $\verifier$ and let $\mathsf{Far}$ be its complement.
    Suppose $\verifier$ sends the challenge state $\ket{\psi_i}=\ket{c_i}_{a_i}$ at
    time $t_i^{\text{send}}$ and accepts the round only if it receives a response
    system by time $t_i^{\text{send}}+\Delta t_{\max}$.

    Then, conditioned on the entire history available at time $t_i^{\text{send}}$
    (and on any internal quantum state), any (classical or quantum) message that
    is generated entirely in $\mathsf{Far}$ and is received by $\verifier$ by the
    deadline $t_i^{\text{send}}+\Delta t_{\max}$ is statistically independent of
    the fresh challenge bit $c_i$.
\end{lemma}

\begin{proof}
Let the message be emitted from some location at distance $\delta>B$ from
$\verifier$. To arrive by $t_i^{\text{send}}+\Delta t_{\max}
=t_i^{\text{send}}+2B/c$, it must be emitted no later than
$t_i^{\text{send}}+2B/c-\delta/c < t_i^{\text{send}}+\delta/c$, i.e., strictly
before any light-speed signal carrying $\ket{\psi_i}$ could reach the emitter.
Therefore the emitted message cannot depend on $c_i$.
\end{proof}

Lemma \ref{lem:oob-general} is a direct consequence of relativity (no superluminal signalling). In the following, we prove the single-round distance-fraud bound and then aggregate across rounds using Lemma \ref{lem:bern-kl}.

\begin{lemma}[Distance-fraud single-round bound]\label{lem:df-single}
  For every round $i$ and any prior transcript, a distant prover's timing-respecting strategy satisfies
  \[
  \Pr[I_i=1] \;\le\; \tfrac12.
  \]
\end{lemma}
  
\begin{proof}
  \noindent\emph{Timing.}\par
  \noindent Since $\prover^\star$ is at distance $d>B$, any message from $\prover^\star$ that reaches $\verifier$ by the deadline cannot depend on the fresh challenge bit $c_i$ (Lemma~\ref{lem:oob-general}). Thus, the prover's returned state $\ket{\phi_i}$ is independent of $c_i$ and fixed by the private randomness $r_i$ of $\prover^\star$. This is the only time-respecting strategy that can be applied by $\prover^\star$.
  
  \smallskip
  \noindent\emph{Value.}\par
  \noindent
  $\verifier$ measures the response in the known basis $b_i$ and accepts the value check iff the outcome equals the uniform challenge bit $c_i$, which is hidden from $\prover^\star$. Writing $\Pi_{b_i,r_i}:=\ket{r_i}_{b_i}\bra{r_i}_{b_i}$, we have
  \[
  \begin{aligned}
  \Pr[I_i=1 \mid \ket{\phi_i}]
    &= \tfrac12\sum_{r_i\in\{0,1\}} \braket{\phi_i | \Pi_{b_i,r_i} | \phi_i}\\
    &= \tfrac12\,\braket{\phi_i | \bigl(\Pi_{b_i,0}+\Pi_{b_i,1}\bigr) | \phi_i}\\
    &= \tfrac12\,\braket{\phi_i | \phi_i}\\
    &= \tfrac12,
  \end{aligned}
  \]
  since $\Pi_{b_i,0}+\Pi_{b_i,1}=\mathbb{I}$. Averaging over $\ket{\phi_i}$ gives $\Pr[I_i=1]\le 1/2$. Sending a response after the deadline can only decrease the success probability. 
\end{proof}

For the multi-round distance-fraud security, we have the following Theorem.

\begin{theorem}[Distance-fraud multi-round bound]\label{lem:df-multi}
  For any timing-respecting strategy and any threshold $\tau$ with $\tau>n/2$,
\[
  \Pr[S \ge \tau]
  \;\le\;
  \exp\Bigl(-\,n \cdot D\bigl(\tfrac{\tau}{n}\,\big\Vert\,\tfrac12\bigr)\Bigr).
\]
\end{theorem}

\begin{proof}
By Lemma~\ref{lem:df-single} we have $\mathbb{E}[I_i \mid \mathcal{F}_{i-1}] \le 1/2$ for the natural filtration capturing the transcript and the adversary's state.
Applying Lemma~\ref{lem:bern-kl} with $p=1/2$ therefore yields the stated bound for every threshold $\tau>n/2$.
\end{proof}

\subsection{Mafia-fraud security}\label{subsec:mafia-fraud-security}

Recall the mafia-fraud experiment defined in \Cref{def:mf-experiment}.
The per-round basis strings $a=(a_1,\ldots,a_n)$ and $b=(b_1,\ldots,b_n)$ are derived from the long-term key $x$ and the public nonces $(N_v,N_p)$ via a quantum-secure PRF. Consequently, they remain computationally indistinguishable from uniform strings for any QPT adversary pair $(\adversary_1,\adversary_2)$ lacking knowledge of~$x$.
In the fast phase, only the adversary $\adversary_1$ (located near the verifier) can respond within the time bound~$B$; $\adversary_2$ may assist only through the (untimed) learning phase.

Abidin's original analysis of this protocol estimated the single-round mafia-fraud success probability at $3/4$~\cite{abidin2019quantum}. However, this underestimates the adversary's advantage. Verschoor later demonstrated that by combining a more capable pre-ask phase with a tailored response strategy, a mafia adversary can achieve a success probability of $7/8$ per round~\cite{verschoor2022quantum}. We summarize Verschoor's attack using our notation below. The core observation is that an adversary can correctly respond to the verifier's challenge simply by knowing whether $a_i$ and $b_i$ belong to the same basis.
For each fast round~$i$, define the parity bit:
\[
  k_i := a_i \oplus b_i \in \{0,1\}.
\]
The verifier's challenge state is $\ket{\psi_i} = \ket{c_i}_{a_i}$, and an honest prover's response is $\ket{\phi_i} = \ket{c_i}_{b_i}$. The honest prover's operation can be characterised as:
\begin{itemize}
  \item if $k_i=0$ (same basis): applying the identity gate (reflecting the challenge), and
  \item if $k_i=1$ (different bases): applying the Hadamard gate (mapping the basis $\{\ket{0}_0,\ket{1}_0\}$ to $\{\ket{0}_1,\ket{1}_1\}$ and vice versa).
\end{itemize}
Thus, knowledge of $k_i$ suffices for an adversary near the verifier to emulate the honest prover's behaviour during the fast phase.

Verschoor's attack is an intra-session \emph{pre-ask} strategy. It exploits the time gap between the slow phase and the fast phase within the \emph{challenge session}.

\begin{enumerate}
    \item \emph{Slow phase relay.} $(\adversary_1, \adversary_2)$ relay the slow phase messages between $\verifier$ and $\prover$. Consequently, all parties agree on the fresh nonces $N_v, N_p$ and derive the same session-specific secrets $a, b$.
    \item \emph{Pre-ask (Extraction).} Before the verifier starts the fast phase, $\adversary_2$ (near $\prover$) initiates a fast phase execution with $\prover$. For each round $i$, $\adversary_2$ sends the intermediate state $\ket{\xi} = \cos(\frac{3\pi}{8})\ket{0} + \sin(\frac{3\pi}{8})\ket{1}$. The honest $\prover$ measures in basis $a_i$ and returns the result. $\adversary_2$ measures this response to obtain a guess $k'_i$ for the parity $k_i = a_i \oplus b_i$, which is correct with probability $3/4$~\cite{verschoor2022quantum}.
    \item \emph{Fast phase challenge.} When $\verifier$ sends the actual challenge $\ket{c_i}_{a_i}$ to $\adversary_1$, $\adversary_1$ uses the pre-computed guess $k'_i$. If $k'_i=1$, they apply a Hadamard gate (basis swap) to the challenge; otherwise, they reflect it.
\end{enumerate}

\begin{lemma}[Mafia-fraud single-round attack~\cite{verschoor2022quantum}]
\label{lem:mf-single}
There exists a mafia-fraud adversary pair $(\adversary_1,\adversary_2)$ in the experiment of \Cref{def:mf-experiment} such that, in each fast round~$i$, the acceptance indicator~$I_i$ satisfies
\[
  \Pr[I_i = 1]
  \;=\;
  \tfrac{3}{4}\cdot 1 \;+\; \tfrac{1}{4}\cdot\tfrac{1}{2}
  \;=\;
  \tfrac{7}{8}.
\]
\end{lemma}

\begin{proof}
\noindent\emph{Timing.}\par
\noindent By definition of the mafia-fraud experiment (\Cref{def:mf-experiment}), $\adversary_1$ is located near $\verifier$ during the fast phase and can therefore return a response within the deadline.

\smallskip
\noindent\emph{Value.}\par
\noindent
Based on the pre-ask analysis above, $\Pr[k'_i = k_i]=3/4$ and $\Pr[k'_i \neq k_i]=1/4$.
If $k'_i=k_i$, $\adversary_1$'s transformation produces exactly $\ket{c_i}_{b_i}$, so the verifier accepts with probability~$1$.
If $k'_i\neq k_i$, the verifier measures in the wrong basis, and the outcome equals $c_i$ with probability~$1/2$.
Conditioning on these two cases yields:
\[
  \Pr[I_i = 1]
    = \Pr[k'_i=k_i]\cdot 1
      + \Pr[k'_i\neq k_i]\cdot\tfrac{1}{2}
    = \tfrac{3}{4} + \tfrac{1}{4}\cdot\tfrac{1}{2}
	    = \tfrac{7}{8}.
\]
\end{proof}

We establish the multi-round security bound in the following.

\begin{theorem}[Mafia-fraud multi-round bound for Verschoor's attack]\label{lem:mf-multi}
  Let $(\adversary_1,\adversary_2)$ be the mafia-fraud adversary pair from Lemma~\ref{lem:mf-single}, and let $S=\sum_{i=1}^n I_i$ be the number of accepted rounds when they execute that strategy independently in each of the $n$ fast rounds. Then, for any threshold $\tau$ with $\tau > \tfrac{7}{8} n$,
  \[
    \Pr[S \ge \tau]
    \;\le\;
    \exp\Bigl(-\,n \cdot D\bigl(\tfrac{\tau}{n}\,\big\Vert\,\tfrac78\bigr)\Bigr).
  \]
\end{theorem}

\begin{proof}
For this fixed adversary pair $(\adversary_1,\adversary_2)$, Lemma~\ref{lem:mf-single} gives $\mathbb{E}[I_i \mid \mathcal{F}_{i-1}] = 7/8$ for every round $i$, where $\mathcal{F}_{i-1}$ is the natural filtration capturing the transcript and the adversary's internal state. Applying Lemma~\ref{lem:bern-kl} with $p=7/8$ therefore yields the stated bound on $\Pr[S \ge \tau]$ for all $\tau > \tfrac{7}{8} n$.
\end{proof}

\subsection{Terrorist-fraud security}\label{sec:terroristfraud}

Recall the terrorist-fraud experiment \Cref{def:tf-experiment}, where a distant dishonest prover $\prover^\star$ collaborates with a nearby helper $\mathcal{A}$ who obeys the distance bound $B$.

The QDB protocol~\cite{abidin2019quantum} is inherently insecure against TF because the fast-phase behaviour of an honest prover is fully determined by the session-specific basis strings $(a,b)$ that are computed in the untimed slow phase. Concretely, once a party knows $(a_i,b_i)$ for a given round~$i$, it can perfectly replicate the honest prover's fast-phase behaviour by measuring the incoming challenge state in basis $a_i$ to recover the encoded bit, and immediately re-encoding that same bit in basis $b_i$ and sending it back. No additional secret information (beyond the bases) is needed during the fast phase.

In the terrorist-fraud experiment \Cref{def:tf-experiment}, the prover itself is dishonest and holds the long-term key $x$. After the public nonces $(N_v,N_p)$ are fixed in the slow phase, $\prover^\star$ can compute the session string $(a,b)=f_x(N_v,N_p)$ and transmit it to the nearby helper~$\mathcal{A}$ before the fast phase begins. During the fast phase, $\mathcal{A}$ is co-located with~$\verifier$ and therefore satisfies the timing bound. Moreover, for each fast round~$i$, upon receiving the challenge $\ket{c_i}_{a_i}$, the helper $\mathcal{A}$ measures in basis $a_i$ and obtains $c'_i$, which equals $c_i$ in the noiseless model, and then prepares and returns $\ket{c'_i}_{b_i}$. This is exactly the honest response behaviour described in~\Cref{subsec:qdb-protocol}, so $\verifier$ accepts each round with probability~$1$ by Lemma~\ref{lem:honest-single} in the noiseless setting.

Crucially, the information shared by $\prover^\star$ with $\mathcal{A}$ can remain \emph{non-transferable} in the sense of~\Cref{def:non-transferable-assistance}. The strings $(a,b)$ are tied to the current session via the fresh nonces and are not reusable across further sessions. In a fresh execution with new nonces $(N_v',N_p')$, the required bases are $(a',b')=f_x(N_v',N_p')$, which cannot be computed from $(a,b)$ without knowing $x$. Under the quantum-secure PRF assumption, revealing $(a,b)$ for one nonce pair does not expose the long-term key $x$ nor enable computing the bases for fresh nonces, so $\mathcal{A}$ alone cannot (except with negligible probability) make $\verifier$ accept in a subsequent execution with freshly sampled nonces. In other words, the protocol does not enforce that any help sufficient to pass the \emph{fast} rounds must also be transferable. This is precisely the structural feature that makes perfect TF collusion possible under~\Cref{def:tf-experiment}.

Hence there exists a pair $(\prover^\star,\mathcal{A})$ with non-transferable assistance whose terrorist-fraud advantage is non-negligible, so the QDB protocol is \emph{not} terrorist-fraud secure in the sense of \Cref{def:tf}.

\subsection{Soundness}\label{sec:soundness}

Recall the definition of soundness from \Cref{def:T-soundness}. Based on the single-round security analyses of the QDB protocol, we establish the following bounds:
\begin{itemize}
  \item An upper bound of $1/2$ on the per-round distance-fraud success probability (Lemma~\ref{lem:df-single});
  \item A mafia-fraud strategy that succeeds with per-round probability $7/8$ (Lemma~\ref{lem:mf-single}); this is currently the best known attack, but we do not know if it is optimal;
  \item The protocol is not terrorist-fraud secure under our model (see~\Cref{sec:terroristfraud}).
\end{itemize}
Consequently, we restrict our focus to \(\{DF,MF\}\)-soundness. While the strategy detailed in Lemma~\ref{lem:mf-single} represents the most effective known mafia-fraud attack against this protocol, neither we nor Verschoor~\cite{verschoor2022quantum} claim that this attack is optimal. To ensure a negligible false-acceptance rate against both distance-fraud and mafia-fraud simultaneously, the acceptance threshold must be set strictly above the maximum per-round cheating probability. We define this threshold using a tunable slack parameter $\varepsilon'$.

\begin{theorem}[Soundness given per-round cheating bounds]
\label{thm:soundness-DFMF}
Let $p_{\mathrm{DF}}^{\max}\le 1/2$ denote the upper bound on the per-round distance-fraud success probability (where $p_{\mathrm{DF}}^{\max}=1/2$ by Lemma~\ref{lem:df-single}), and let $p_{\mathrm{MF}}^{\max}<1$ be any upper bound on the per-round mafia-fraud success probability for the adversary class under consideration. Fix any $\varepsilon'>0$ such that $\max(p_{\mathrm{DF}}^{\max},p_{\mathrm{MF}}^{\max})+\varepsilon'<1$ and set the (integer) acceptance threshold
\[
\tau := \left\lceil n\bigl(\max(p_{\mathrm{DF}}^{\max},p_{\mathrm{MF}}^{\max})+\varepsilon'\bigr)\right\rceil.
\]
Then, letting $S=\sum_{i=1}^{n} I_i$ denote the number of accepted rounds:
\[
\varepsilon_{\mathrm{DF}}
= \Pr\bigl[S\ge\tau \mid \text{DF}\bigr]
\;\le\;
\exp\Bigl(-\,n \cdot D\bigl(\tfrac{\tau}{n}\,\big\Vert\,p_{\mathrm{DF}}^{\max}\bigr)\Bigr),
\]
\[
\varepsilon_{\mathrm{MF}}
= \Pr\bigl[S\ge\tau \mid \text{MF}\bigr]
\;\le\;
\exp\Bigl(-\,n \cdot D\bigl(\tfrac{\tau}{n}\,\big\Vert\,p_{\mathrm{MF}}^{\max}\bigr)\Bigr).
\]
In particular, if $n=\poly(\lambda)$ and $p_{\mathrm{MF}}^{\max}<1$ is constant, then both $\varepsilon_{\mathrm{DF}}$ and $\varepsilon_{\mathrm{MF}}$ are negligible in $\lambda$. Since $p_{\mathrm{MF}}^{\max}\ge p_{\mathrm{DF}}^{\max}$ in our setting, it suffices in practice to determine $n$ and $\tau$ based on $p_{\mathrm{MF}}^{\max}$.
\end{theorem}

Applying \Cref{thm:soundness-DFMF} to the QDB protocol, we use $p_{\mathrm{DF}}^{\max}=1/2$ from Lemma~\ref{lem:df-single}. Regarding mafia-fraud, the attack by Verschoor~\cite{verschoor2022quantum} achieves a per-round success probability of $7/8$ (Lemma~\ref{lem:mf-single}). Instantiating Theorem~\ref{thm:soundness-DFMF} with $p_{\mathrm{MF}}^{\max}=7/8$ yields explicit bounds on the false-accept probability against this specific strategy and any other attack with a success rate of at most $7/8$. Should a more effective attack be discovered, the parameters can be adjusted by substituting the updated success probability into Lemma~\ref{lem:bern-kl}.

To ensure completeness in the presence of noise, the expected honest per-round acceptance probability $p(\eta)$ from \Cref{thm:completeness-noise} must strictly exceed the threshold ratio $\tau/n$. Specifically, an honest execution is accepted with probability $1-\negl{\lambda}$ provided that:
\(
p(\eta)=1-\eta+\eta^2/2 > \tau/n.
\)
Equivalently, for a given threshold fraction $u:=\tau/n\in(\max(p_{\mathrm{DF}}^{\max},p_{\mathrm{MF}}^{\max}),1)$ (where $u\in(7/8,1)$ for our instantiation), the depolarizing parameter must satisfy:
\[
\eta \;<\; 1 - \sqrt{\,2u-1\,}.
\]

Figure~\ref{fig:tradeoff} illustrates the performance trade-offs for the QDB protocol: (a) soundness, showing the number of rounds $n$ required to achieve $P_{\mathrm{FA}}\le 2^{-80}$ across varying threshold ratios $u=\tau/n$; and (b) completeness, showing the maximum tolerable depolarizing noise $\eta_{\max}$ per hop.

\begin{figure}[t]
  \centering
  \begin{subfigure}{0.48\linewidth}
    \centering
    \includegraphics[width=\linewidth]{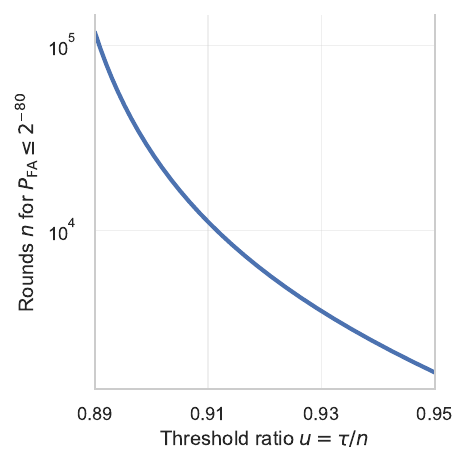}
    \caption{Rounds $n$ required for soundness ($P_{\mathrm{FA}}\le 2^{-80}$).}
    \label{fig:tradeoff-a}
  \end{subfigure}\hfill
  \begin{subfigure}{0.48\linewidth}
    \centering
    \includegraphics[width=\linewidth]{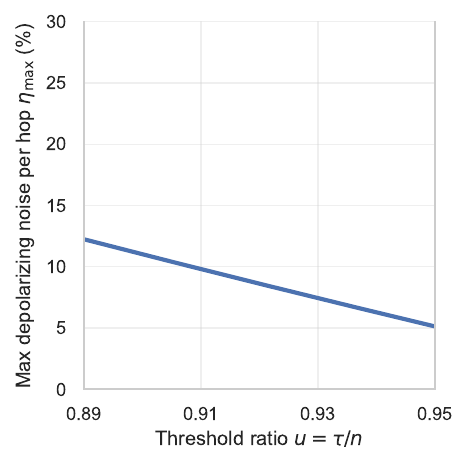}
    \caption{Maximum depolarizing noise $\eta_{\max}$ for completeness.}
    \label{fig:tradeoff-b}
  \end{subfigure}
  \caption{Performance trade-offs as a function of the threshold ratio $u=\tau/n$: (a) soundness; (b) completeness.}
  \label{fig:tradeoff}
\end{figure}

\subsection{Towards terrorist-fraud resistance}

The TF attack described in~\Cref{sec:terroristfraud} succeeds because a nearby helper can reproduce the complete
fast-phase behaviour given only session-specific information that the
dishonest prover can disclose before the fast phase begins. These directions should be understood as structural redesign options rather than simple local patches to the protocol of~\cite{abidin2019quantum}. Preventing such
collusion therefore requires that any assistance sufficient to answer the fast
rounds either (i) cannot be provided without violating the distance bound, or
(ii) necessarily enables the helper to impersonate the prover in future sessions, and is thus excluded by the non-transferability requirement of~\Cref{def:non-transferable-assistance}.

A common design direction in classical DB protocols is to bind the fast-phase responses to long-term secret material, so that producing correct responses requires access to a reusable secret (and in such protocol designs this can enable extraction). Some classical DB protocols implement this principle by coupling the responses directly to the prover's long-term key, so that enabling an accomplice to answer reliably entails disclosing enough information to recover it~\cite{bussard2005distance}.

A complementary direction is to increase the amount of fresh challenge entropy that must be handled during the fast phase, for instance by using multi-bit or multi-state challenges, and to make the expected response depend on this larger challenge space. Intuitively, if correct fast-phase behaviour corresponds to a keyed response function over a larger challenge space, then enabling a helper
without revealing the long-term secret may require providing a large amount of
challenge-dependent information (in the extreme, an explicit response table)
whose size grows with the challenge space, at the cost of additional fast-phase
bandwidth and processing. Classical DB protocols explore this trade-off by extending binary challenges to larger challenge spaces~\cite{avoine2009rfid}.

A third, system-level mitigation is to execute the key-dependent parts of the protocol inside tamper-resistant hardware on the prover side, such that the long-term key $x$ and any session-derived secrets, including the basis strings $(a,b)$, are not readable or exfiltratable by the holder of the device. In this setting, the TF attack is blocked at its source because even a malicious prover cannot leak $(a,b)$ to a nearby helper ahead of the fast phase, since it cannot access them from the secure hardware. This hardware assumption has been noted in the DB literature as a practical way to obtain robustness against TF~\cite{singelee2007distance}.

Developing concrete protocol modifications is left for future work. Our goal is to identify the underlying cause of TF insecurity in~\cite{abidin2019quantum} under the model of~\Cref{def:tf-experiment}. The fast phase can be entirely reproduced from session-derived secrets computed in the slow phase, allowing a dishonest prover to enable a nearby helper without revealing the long-term key.

\section{Conclusion}\label{sec:conclusion}

In this work we proposed a reusable, game-based security framework for QDB protocols and applied it to the QDB protocol~\cite{abidin2019quantum}. The framework adapts the classical methodology~\cite{boureanu2015practical} to the quantum setting by making explicit the quantum-capable adversary model and the formal distance-, mafia-, and terrorist-fraud experiments, fixing a simple noise model, and isolating the per-round success probabilities that drive the overall security guarantees. Within this setting we proved completeness of the protocol in the presence of depolarizing noise by characterising the honest per-round acceptance probability and lifting it to the multi-round setting; for active adversaries we proved an upper bound of~$1/2$ on the per-round distance-fraud success probability, and showed that the best known mafia-fraud strategy succeeds with probability~$7/8$ per round. Plugging these per-round bounds into our generic Lemma~\ref{lem:bern-kl} yields explicit multi-round soundness bounds (Theorem~\ref{thm:soundness-DFMF}) that yield negligible false-accept probability for suitable choices of~$n$ and~$\tau$. Our analysis also shows that the protocol is not terrorist-fraud secure under our model (as established in \Cref{sec:terroristfraud}); we therefore treat terrorist-fraud resistance as a future research direction rather than claiming a drop-in fix.

Conceptually, the framework separates protocol-specific details (such as the
encoding and message flow) from reusable security ingredients (adversary
experiments, noise and timing assumptions, and concentration bounds). This
makes it straightforward to compare different QDB proposals on a common
basis.

We conclude by highlighting several promising directions that remain open:
\begin{itemize}
    \item
    \emph{Beyond i.i.d.\ depolarizing noise.}
    We adopt a simple, protocol-agnostic noise model in~\Cref{subsec:noise-model} as a baseline for completeness and parameter sizing. 
    Realistic implementations of QDB protocols will also exhibit photon loss and erasures, detector inefficiency, dark counts and background clicks, basis-dependent misalignment (different error rates in $Z$ vs.\ $X$), and potentially time-varying or correlated noise across rounds.
    A natural extension is to replace the simple depolarizing model with more realistic per-round noise and an explicit detection and erasure model, and derive the corresponding acceptance probabilities and protocol parameters accordingly.
    In particular, any treatment of \emph{inconclusive} rounds, such as those caused by photon losses, must be modelled carefully, since an adversary could try to induce selective losses on hard rounds and answer only when confident.

    \item \emph{Terrorist-fraud resistance.}
    Our analysis shows that the QDB protocol~\cite{abidin2019quantum} is
    vulnerable to terrorist-fraud. A natural next step is to design and
    analyse QDB protocols that achieve full DF/MF/TF soundness within the
    same framework.
    
    \item \emph{Multiple parties.}
    We focused on the single-verifier, single-prover setting. Extending the framework to multiple verifiers and/or multiple provers is an important direction for making QDB applicable to realistic localisation systems.
    
    \item \emph{Experimental validation.}
    Validating the
    resulting bounds in a full experimental implementation on
    state-of-the-art hardware would further bridge the gap between the
    theoretical guarantees and real-world deployments.
\end{itemize}

Addressing these points will strengthen the case for QDB as a practical
building block for secure localisation and access control, and our framework
is intended to serve as a starting point for such follow-up work.

\bibliographystyle{unsrt}
\bibliography{references}

@inproceedings{brands1993distance,
  title={{Distance-Bounding Protocols}},
  author={Brands, Stefan and Chaum, David},
  booktitle={Workshop on the Theory and Application of Cryptographic Techniques},
  pages={344--359},
  year={1993},
  organization={Springer}
}

@inproceedings{hancke2005rfid,
  title={An {RFID} {Distance Bounding} Protocol},
  author={Hancke, Gerhard P and Kuhn, Markus G},
  booktitle={First international conference on security and privacy for emerging areas in communications networks (SECURECOMM'05)},
  pages={67--73},
  year={2005},
  organization={IEEE}
}

@article{hancke2011design,
  title={{Design of a Secure Distance-Bounding Channel for {RFID}}},
  author={Hancke, Gerhard P},
  journal={Journal of Network and Computer Applications},
  volume={34},
  number={3},
  pages={877--887},
  year={2011},
  publisher={Elsevier}
}

@inproceedings{abidin2016towards,
  title={Towards {Quantum} {Distance Bounding} Protocols},
  author={Abidin, Aysajan and Marin, Eduard and Singel{\'e}e, Dave and Preneel, Bart},
  booktitle={International Workshop on Radio Frequency Identification: Security and Privacy Issues},
  pages={151--162},
  year={2016},
  organization={Springer}
}

@inproceedings{abidin2019quantum,
  title={{Quantum Distance Bounding}},
  author={Abidin, Aysajan},
  booktitle={Proceedings of the 12th Conference on Security and Privacy in Wireless and Mobile Networks},
  pages={233--238},
  year={2019}
}

@inproceedings{abidin2024entanglement,
  title={{Entanglement-based Mutual Quantum Distance Bounding}},
  author={Abidin, Aysajan and Eldefrawy, Karim and Singel{\'e}e, Dave},
  booktitle={International Symposium on Cyber Security, Cryptology, and Machine Learning},
  pages={219--235},
  year={2024},
  organization={Springer}
}

@inproceedings{bogner2024entangled,
  title={{Entangled States and Bell's Inequality: A New Approach to Quantum Distance Bounding}},
  author={Bogner, Kevin and Singel{\'e}e, Dave and Abidin, Aysajan},
  booktitle={2024 IEEE Symposium on Computers and Communications (ISCC)},
  pages={1--6},
  year={2024},
  organization={IEEE}
}

@inproceedings{bogner2025continuous,
  title={{Continuous Variable Quantum Distance Bounding}},
  author={Bogner, Kevin and Abidin, Aysajan and Singel{\'e}e, Dave},
  booktitle={IEEE INFOCOM 2025-IEEE Conference on Computer Communications Workshops (INFOCOM WKSHPS)},
  pages={1--6},
  year={2025},
  organization={IEEE}
}

@article{boureanu2015practical,
  title={{Practical and Provably Secure Distance-Bounding}},
  author={Boureanu, Ioana and Mitrokotsa, Aikaterini and Vaudenay, Serge},
  journal={Journal of Computer Security},
  volume={23},
  number={2},
  pages={229--257},
  year={2015},
  publisher={SAGE Publications Sage UK: London, England}
}

@article{wootters1982single,
  title={{A Single Quantum Cannot be Cloned}},
  author={Wootters, William K and Zurek, Wojciech H},
  journal={Nature},
  volume={299},
  number={5886},
  pages={802--803},
  year={1982},
  publisher={Nature Publishing Group UK London}
}

@inproceedings{bennett1984bb84,
  title        = {{Quantum Cryptography: Public Key Distribution and Coin Tossing}},
  author       = {Bennett, Charles H. and Brassard, Gilles},
  booktitle    = {Proceedings of IEEE International Conference on Computers, Systems and Signal Processing},
  address      = {Bangalore, India},
  pages        = {175--179},
  year         = {1984},
  organization = {IEEE},
  note         = {Reprinted in \emph{Theoretical Computer Science} 560 (2014), 7--11}
}

@article{zhandry2021construct,
  title={{How to Construct Quantum Random Functions}},
  author={Zhandry, Mark},
  journal={Journal of the ACM (JACM)},
  volume={68},
  number={5},
  pages={1--43},
  year={2021},
  publisher={ACM New York, NY}
}

@phdthesis{verschoor2022quantum,
  title={{Quantum Information in Security Protocols}},
  author={Verschoor, Sebastian Reynaldo},
  year={2022},
  school={University of Waterloo}
}

@book{nielsen2010quantum,
  title     = {{Quantum Computation and Quantum Information: 10th Anniversary Edition}},
  author    = {Nielsen, Michael A. and Chuang, Isaac L.},
  publisher = {Cambridge University Press},
  address   = {Cambridge, UK},
  year      = {2010}
}

@book{mitzenmacher2017probability,
  title={{Probability and Computing: Randomization and Probabilistic Techniques in Algorithms and Data Analysis}},
  author={Mitzenmacher, Michael and Upfal, Eli},
  year={2017},
  publisher={Cambridge University Press},
  address={Cambridge, UK}
}

@article{buhrman2014position,
  title={{Position-based Quantum Cryptography: Impossibility and Constructions}},
  author={Buhrman, Harry and Chandran, Nishanth and Fehr, Serge and Gelles, Ran and Goyal, Vipul and Ostrovsky, Rafail and Schaffner, Christian},
  journal={SIAM Journal on Computing},
  volume={43},
  number={1},
  pages={150--178},
  year={2014},
  publisher={SIAM}
}

@inproceedings{bussard2005distance,
  title={{Distance-bounding Proof of Knowledge to Avoid Real-Time Attacks}},
  author={Bussard, Laurent and Bagga, Walid},
  booktitle={IFIP international information security conference},
  pages={223--238},
  year={2005},
  organization={Springer}
}

@inproceedings{avoine2009rfid,
  title={{RFID Distance Bounding Multistate Enhancement}},
  author={Avoine, Gildas and Floerkemeier, Christian and Martin, Benjamin},
  booktitle={International conference on cryptology in India},
  pages={290--307},
  year={2009},
  organization={Springer}
}

@inproceedings{singelee2007distance,
  title={{Distance Bounding in Noisy Environments}},
  author={Singel{\'e}e, Dave and Preneel, Bart},
  booktitle={European workshop on security in ad-hoc and sensor networks},
  pages={101--115},
  year={2007},
  organization={Springer}
}

@article{abidin2020detecting,
  title={{On Detecting Relay Attacks on RFID Systems using Qubits}},
  author={Abidin, Aysajan},
  journal={Cryptography},
  volume={4},
  number={2},
  pages={14},
  year={2020},
  publisher={MDPI}
}

\end{document}